\documentclass[conference,10pt]{IEEEtran}
\usepackage{cite}

\usepackage{graphicx}
\usepackage{ifpdf}
\ifpdf
  \usepackage{epstopdf}
\fi

\usepackage[cmex10]{amsmath}
\usepackage{amsthm}
\usepackage{amssymb}
\usepackage{cases}
\usepackage{bm}

\usepackage[caption=false,font=footnotesize]{subfig}

\usepackage{fixltx2e}

\usepackage{url}

\usepackage{algorithm}
\usepackage{algorithmic}
\begin{document}
\title{How Often Should CSI be Updated for Massive MIMO Systems with Massive Connectivity?}
\author{\IEEEauthorblockN{Ruichen Deng, Zhiyuan Jiang, Sheng Zhou, and Zhisheng Niu}
        \IEEEauthorblockA{Tsinghua National Laboratory for Information Science and Technology\\
        Department of Electronic Engineering, Tsinghua University, Beijing 100084, China \\
        Email: drc13@mails.tsinghua.edu.cn,  \{zhiyuan,sheng.zhou,niuzhs\}@tsinghua.edu.cn\\}}

\maketitle
\theoremstyle{plain}
\newtheorem{theorem}{Theorem}
\theoremstyle{remark}
\newtheorem{re:re1}{Remark}

\begin{abstract}
  Massive multiple-input multiple-output (MIMO) systems need to support massive connectivity for the application of the Internet of things (IoT). The overhead of channel state information (CSI) acquisition becomes a bottleneck in the system performance due to the increasing number of users. An intermittent estimation scheme is proposed to ease the burden of channel estimation and maximize the sum capacity. In the scheme, we exploit the temporal correlation of MIMO channels and analyze the influence of the age of CSI on the downlink transmission rate using linear precoders.  We show the CSI updating interval should follow a quasi-periodic distribution and reach a trade-off between the accuracy of CSI estimation and the overhead of CSI acquisition by optimizing the CSI updating frequency of each user. Numerical results show that the proposed intermittent scheme provides significant capacity gains over the conventional continuous estimation scheme.
\end{abstract}
\IEEEpeerreviewmaketitle

\section{Introduction}
By applying a large number of antenna elements, the massive multiple-input multiple-output (MIMO) system exploits the potentials of spatial multiplexing with tremendous degrees-of-freedom, and is believed to increase the spectral efficiency 10 times or more as well as improving the radiated energy efficiency in the order of 100 times \cite{MassiveMIMOOverview}. Exploiting the advantages of massive MIMO technique requires the knowledge of channel state information (CSI) at the base station (BS) side. The canonical massive MIMO protocol is to operate in time-division-duplex (TDD) mode in order to acquire CSI through channel reciprocity \cite{TenMyths}. However, the overhead of channel estimation is still unbearable for the application of the Internet of things (IoT) \cite{3GPPM2M}, where the systems need to serve a vast quantity of users within the limited channel coherence time.

Many prior works try to exploit the spatial correlation of channels to reduce the overhead of channel acquisition. The spatial correlation matrices of users are used to design inter-cell and intra-cell pilot reuse schemes \cite{Yin13}\cite{PilotReuse}. Ref. \cite{CompressiveSensing} assumes channel sparsity in spatial domain and estimates channels based on the compressive sensing theory. Mutual-information-optimal pilots for minimum mean square error (MMSE) estimation are proposed in \cite{Jiang15,Deng16,jiang17cl}. 

Another aspect of the channel statistics worth exploring is the temporal correlation. The correlation between the current channel and the previous channel becomes weaker with longer time elapsed \cite{Abdi02}. Under the assumption of the Gilbert-Elliot channel model, the authors introduce the concept \emph{the age of CSI} and study its impact on a general utility function \cite{Costa15}. A more realistic model to characterize the temporal correlation of channels is the first-order Gaussian-Markov process \cite{Dong04}. Training techniques and pilot beam design for single user are discussed under this model \cite{Choi14}\cite{PilotBeamPatternDesign}. In \cite{4,8}, the authors study the multiuser dynamic channel acquisition problem under the block fading channel.

Our work focuses on the multiuser scenario. To ease the burden of CSI acquisition, we propose an intermittent estimation scheme. Different from the conventional scheme which continuously estimates channels of the whole users in every channel block, the proposed scheme exploits the temporal correlation and requires users to estimate their channels intermittently among blocks. The CSI updating interval of each user is shown to follow a quasi-periodic distribution and the CSI updating frequency is optimized based on the temporal correlation coefficient of each user. 
 Hence, the scheme reaches a trade-off between the accuracy of the CSI and the overhead of the CSI acquisition to achieve the maximum sum capacity.

The rest of the paper is organized as follows. Section II describes the system model adopted in this paper. The influence of aged CSI on the downlink transmission rate is investigated in Section III. Then we propose the intermittent estimation scheme in Section IV. The numerical results are presented in Section V and conclusions are drawn in Section VI.

\section{System Model}
We consider a massive MIMO system with $M$ antenna elements serving $K$ users. The  processes of channel estimation and data transmission occur periodically in a block of $L$ channel uses.  The system operates in the calibrated time-duplex-division (TDD) mode so that channel reciprocity is utilized for channel estimation. More specifically, in the beginning of every block, a number of $T$ channel uses is spent by the base station (BS) on estimating uplink channels, afterwards the BS obtains downlink channel vectors as the transpose of uplink channel vectors, and transmits precoded data in the remaining $L-T$ channel uses.

The uplink channel of the $k$-th user in the $i$-th block is 
\begin{equation}
\bm{g}_k(i)=\sqrt{\beta_k} \bm{h}_k(i),
\end{equation}
where $\sqrt{\beta_k}$ and $\bm{h}_k(i)$ denotes the large and small scale fading coefficient, respectively. The small scale fading coefficient satisfies the complex circularly-symmetric Gaussian distribution $\mathcal{CN}(0,\bm{I}_M)$. It stays constant in each block and evolves between blocks according to the first-order stationary Gaussian-Markov process \cite{Dong04}
\begin{equation}\label{GM process}
\bm{h}_k(i+1)=\rho_k\bm{h}_k(i)+\sqrt{1-\rho_k^2}\bm{e}_k(i),~~i\geq 0,
\end{equation}
where $\rho_k$ is the temporal correlation coefficient of channels between blocks, and $\bm{e}_k(i)$ is the innovation process distributed as $\mathcal{CN}(0,\bm{I}_M)$. We assume the initial channels $\bm{g}_k(0)$ and the innovation processes $\bm{e}_k(i)$ of each user are independent, so the channels of different users are independent in each block.

 If the channel is updated in the $i$-th block, then the channel after $n$ blocks can be easily derived by iterating (\ref{GM process}), which is
\begin{equation}
\bm{h}_k(i+n)=\rho_k^n\bm{h}_k(i)+\sqrt{1-\rho_k^{2n}}\bm{e}_{k,n}(i),~~n\geq 1,
\end{equation}
where $\bm{e}_{k,n}(i)\sim \mathcal{CN}(0,\bm{I}_M)$ is the equivalent innovation process in $n$ blocks. \emph{The age of CSI} \cite{Costa15} is defined as the number of blocks elapsed since last channel estimation in our work. We focus on the impact of the CSI updating frequency and only consider estimation errors caused by aging of CSI. As the age of CSI $n$ increases, the temporal correlation between the estimated channel $\bm{h}_k(i)$ and the channel $\bm{h}_k(i+n)$ of ground truth declines exponentially as $\rho_k^n$, making the estimation less accurate.

\section{Transmission Rate with Aged CSI}
The BS needs to precode the user data before downlink transmission to support spatial multiplexing. The received signal of the $k$-th user is
\begin{equation}
\begin{split}
y_k &= \sqrt{\epsilon_k}\bm{g}_k^T \bm{v}_k d_k+\sum_{i=1,i\neq k}^{K}\sqrt{\epsilon_i}\bm{g}_i^T \bm{v}_i d_i+\bm{n}_k\\
&= \sqrt{\epsilon_k\beta_k}\bm{h}_k^T \bm{v}_k d_k+\sum_{i=1,i\neq k}^{K}\sqrt{\epsilon_i\beta_i}\bm{h}_i^T \bm{v}_i d_i+\bm{n}_k,
\end{split}
\end{equation}
where $\epsilon_k$, $\bm{v}_k$, $d_k$ and $\bm{n}_k$ denotes SNR at the transmitter, precoding vector, data symbol and normalized downlink noise, respectively. A lower bound of per user SINR at the receiver can be obtained by using the similar technique as \cite[Theorem 1]{Jose11}, which is denoted as
\begin{equation}\label{SINR}
\gamma_k \!=\! \frac{\epsilon_k\beta_k|\mathbb{E}\bm{h}_k^T\bm{v}_k|^2}{1\!+\!\epsilon_k\beta_k(\mathbb{E}|\bm{h}_k^T\bm{v}_k|^2\!-\!|\mathbb{E}\bm{h}_k^T\bm{v}_k|^2)\!+\!\sum\limits_{i\neq k}\epsilon_i\beta_i\mathbb{E}|\bm{h}_k^T\bm{v}_i|^2}.
\end{equation}
And the downlink transmission rate is obtained:
\begin{equation}
R_k = \log\left\{1+\gamma_k\right\}.
\end{equation}

We consider two kinds of linear precoding schemes, namely matched filter (MF) and zero forcing (ZF). Assume the estimated channel matrix is $\bm{\hat{H}}_k$, then the precoding matrix of matched filter is
\begin{equation}
\bm{V}^\textrm{MF}_k=\frac{1}{\sqrt{\eta^\textrm{MF}_k}}\bm{\hat{H}}_k^*,
\end{equation}
where $\eta_k^\textrm{MF}$ is the normalization coefficient. In the massive MIMO scenario, the value of $\eta_k^\textrm{MF}$ tends to be constant due to the channel hardening effect, which is derived as
\begin{equation}
\eta^\textrm{MF}_k = \frac{1}{K}\mathbb{E}tr\{\bm{\hat{H}}_k^T\bm{\hat{H}}_k^*\}=M.
\end{equation}

For zero forcing precoding, the precoding matrix becomes
\begin{equation}
\bm{V}^\textrm{ZF}_k=\frac{1}{\sqrt{\eta^{ZF}_k}}\bm{\hat{H}}_k^*(\bm{\hat{H}}_k^T\bm{\hat{H}}_k^*)^{-1},
\end{equation}
where the normalization coefficient is
\begin{equation}
\eta^\textrm{ZF}_k = \frac{1}{K}\mathbb{E}tr\{(\bm{\hat{H}}_k^T\bm{\hat{H}}_k^*)^{-1}\}=\frac{1}{M-K}.
\end{equation}

\begin{theorem}
	If the CSI of the $k$-th user has an age of $n$ blocks, then the SINR using MF precoding is given by
	\begin{equation}
	\gamma_{k}^{\textrm{MF}}(n)=\frac{\epsilon_k\beta_k M\rho_k^{2n}}{1+\sum_{i=1}^{K}\epsilon_i\beta_i},
	\end{equation}
	and the SINR using ZF precoding is given by
	\begin{equation}
	\gamma_{k}^\textrm{ZF}(n)=\frac{\epsilon_k\beta_k (M-K)\rho_k^{2n}}{1+(1-\rho_k^{2n})\sum_{i=1}^{K}\epsilon_i\beta_i}.
	\end{equation}
\end{theorem}
 \begin{proof}
 	See Appendix \ref{appendices1}.
 \end{proof}

From the above theorem, we observe the fact that for both precoding schemes, the SINR functions of the $k$-th user monotonically decreases as its age of CSI grows, and are not affected by the ages of CSI of other users. For the convenience of expressions, we will use the unified term $\gamma_k(n), R_k(n)$ for the two precoding schemes in the rest of the paper.

\section{Intermittent Estimation Scheme}
\subsection{Scheme Description}
The channel estimation overhead increases linearly with the total number of users. When user number increases to a certain amount, the channel block can no longer support the estimation overhead of all users. So we propose an intermittent channel estimation scheme. In every block, the BS chooses a subset of users instead of the whole users to estimate their channels. Afterwards, the BS forms the precoder matrix using the latest CSI of each user and then transmits data to all users. The scheme reduces the overhead of channel estimation at the cost of a decline of the estimation accuracy.
\begin{figure}[!t]
	\centering
	\includegraphics[width=3.3in]{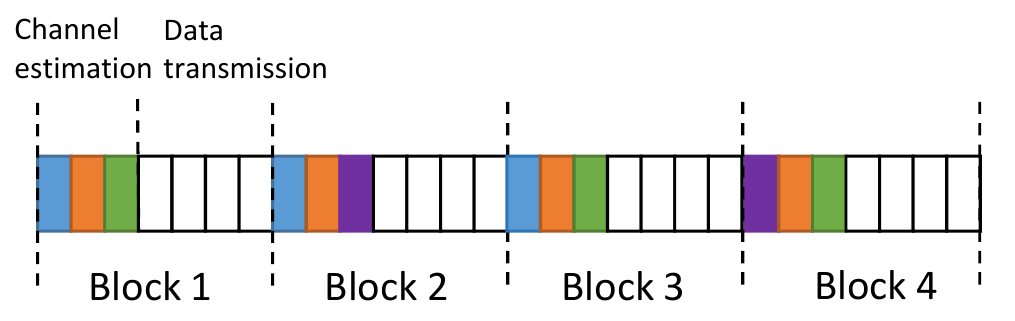}
	\caption{An example of the intermittent estimation scheme.}
	\label{fig:IES}
\end{figure}

An example of the intermittent estimation scheme is illustrated in Fig. \ref{fig:IES}. The system provides data service for total 4 users. The pilot pattern has a period of 4 blocks, and the system selects 3 users from the whole user set for channel estimation in each block. In each period of the pilot pattern, the CSI of blue user is updated twice with an interval of 1 block and once with an interval of 2 blocks. So its average updating interval is $4/3$ blocks and CSI updating frequency equals $3/4$ per block. Similarly, the CSI updating frequencies of orange, green and purple users are $1, 3/4$ and $1/2$ per block, respectively.
 
For a given user set, we try to maximize the average user sum rate by choosing the proper user subset for channel estimation in every block. The optimization problem is formulated as
\begin{equation}\label{original_OP}
\begin{split}
\max_{q_k(i),T}~~ &\lim_{I\rightarrow\infty}\frac{1}{I}\left(1-\frac{T}{C}\right)\sum_{i=0}^{I-1}\sum_{k=1}^{K}R_{k,i}\\
s.t.~~ &q_k(i)\in\{0,1\},~~k=1,\cdots,K\\
&\sum_{k} q_k(i)=T,
\end{split}
\end{equation} 
where $q_k(i)$ indicates whether the $k$-th user is chosen to do channel estimation in the $i$-th block. The sum rate is multiplied by a discounting factor $1-T/C$ where $T/C$ represents the ratio of estimation overhead.

Since the transmission rate of the $k$-th user is only affected by its own age of CSI, rearranging the order of CSI updating intervals has no influence on the rate performance of the $k$-th user. Denote $p_k$ as the CSI updating frequency and $f_{k,n}$ as the distribution CSI updating interval. The CSI updating frequency is the reciprocal of the average CSI updating interval:
\begin{equation}
p_k = \frac{1}{\sum_n n f_{k,n}}.
\end{equation}

We consider a new optimization problem as
\begin{equation}\label{OP}
\begin{split}
\max_{T,p_k,f_{k,n}}~~ &\left(1-\frac{T}{C}\right)\sum_{k=1}^{K}\mathbb{E}\{R_k\}\\
s.t.~~&0\leq f_{k,n}\leq 1,~~k=1,\cdots,K\\
&\sum_n f_{k,n}=1,~~k=1,\cdots,K\\
&\sum_{k} p_k=T,
\end{split}
\end{equation}
which is actually a relaxation of the original problem (\ref{original_OP}). The second constraint is relaxed from pilot length in every block equaling $T$ to average pilot length equaling $T$. When the number of users goes to infinity, the pilot length tends to be the sum of the CSI updating frequencies of all users by the law of large numbers. Therefore, the relaxation is asymptotically tight.

The relaxed problem can be solved in three steps. Firstly, we fix the value of the CSI updating frequency $p_k$ and the pilot length $T$ and try to optimize the distribution  $f_{k,n}$ of CSI updating interval for each user. Then $p_k$ is optimized under fixed value of $T$. Finally we find the optimal $T$ by a one-dimensional search.
\subsection{Optimizing the Distribution of CSI Updating Interval}
To solve the relaxed problem (\ref{OP}), we firstly consider the rate performance of a single user, and optimize $f_{k,n} $ under fixed value of the CSI updating frequency $p_k$. It is formulated as the following optimization problem,
\begin{equation}\label{SU_OP}
\begin{split}
\max_{f_{k,n}}~~ &\mathbb{E}\{R_k\}\\
s.t.~~&0\leq f_{k,n}\leq 1,~~k=1,\cdots,K\\
&\sum_n f_{k,n}=1\\
&\sum_n nf_{k,n}=\frac{1}{p_k}.
\end{split}
\end{equation}

Assume the $k$-th user has done $N_k$ times of channel estimations in a total of $N$ blocks, and there are $F_{k,n}$ CSI updating intervals with block length of $n$. The total rate in a CSI updating interval with block length of $n$ can be calculated as 
\begin{equation}
G_k(n)\triangleq\sum_{i=0}^{n-1}R_k(i),
\end{equation} 
where $R_k(i)=\log(1+\textrm{SINR}_k(i))$ is the transmission rate with a CSI age of $i$ blocks. Obviously, $R_k(i)$ is a non-increasing function with respect to $i$. According to the definition of $p_k$ and $f_{k,n}$, the limitations $\lim_{N\rightarrow\infty} N_k/N=p_k$ and $\lim_{N\rightarrow\infty} F_{k,n}/N_k=f_{k,n}$ both hold. Therefore, the average transmission rate of the $k$-th user is
\begin{equation}
\mathbb{E}\{R_k\} = \lim_{N\rightarrow\infty} \frac{1}{N}\sum_n F_{k,n} G_k(n)=p_k\sum_n f_{k,n}G_k(n).
\end{equation}

We extend $G_k(n)$ to a piecewise function:
\begin{equation}
G_k(x)\! = \!
\begin{cases}
0 & x=0\\
\sum_{i=0}^{x-1}R_k(i)& x\in\mathcal{N}_+\\
(\!1\!-\!x\!+\!\lfloor x \rfloor\!)\!G_k\!(\!\lfloor x \rfloor\!)\!+\!(\!x\!-\!\lfloor x \rfloor\!)\!G_k\!(\!\lfloor x \rfloor\!+\!1)& \!x\!\notin\!\mathcal{N},x\!\geq\! 0
\end{cases}
\end{equation}

The optimal distribution of CSI updating interval is obtained according to the following theorem. 
\begin{theorem}
	The transmission rate of the $k$-th user achieves the maximum value
	\begin{equation}
	S_k(p_k)=p_k G_k(\frac{1}{p_k}),
	\end{equation} 
	if the distribution of CSI updating interval is
	\begin{equation}\label{fkn}
	f_{k,n}=
	\begin{cases}
	1-\frac{1}{p_k}+\lfloor \frac{1}{p_k} \rfloor& n=\lfloor \frac{1}{p_k} \rfloor\\
	\frac{1}{p_k}-\lfloor \frac{1}{p_k} \rfloor& n=\lfloor \frac{1}{p_k} \rfloor+1\\
	0. & else.
	\end{cases}
	\end{equation}
\end{theorem}
\begin{proof}
See Appendix \ref{appendices2}.
\end{proof}
The theorem requires the CSI updating interval to be quasi-periodic. More specifically, if $1/p_k$ is an integer, the CSI should be updated every $1/p_k$ blocks. Otherwise the CSI should be updated every $\lfloor 1/p_k \rfloor$ or $\lfloor 1/p_k \rfloor+1$ blocks with an average updating interval of $1/p_k$ blocks.

\subsection{Optimizing the CSI Updating Frequency}
After optimizing $f_{k,n}$, the transmission rate of the $k$-th user is $S_k(p_k)$. We need to allocate the total estimation resource of $T$ channel uses to all users by optimizing the CSI updating frequency. The problem is formulated as
\begin{equation}\label{newOP}
\begin{split}
\max_{p_k}~~ &\left(1-\frac{T}{C}\right)\sum_{k=1}^{K}S_k(p_k)\\
s.t.~~&0\leq p_k\leq 1,~~k=1,\cdots,K\\
&\sum_{k} p_k=T.
\end{split}
\end{equation}
\begin{theorem}
	The problem (\ref{newOP}) is a concave optimization problem.
\end{theorem}
\begin{proof}
See Appendix \ref{appendices3}.
\end{proof}

The concavity guarantees the global optimality of any local maximum we find. However, the problem is still hard to solve since the function $G_k(x)$ is non-differentiable. So we consider to approximate $G_k(x)$ by replacing its function value by a parabola in each interval $(n-\delta,n+\delta),n\in\mathcal{N}_+$. The approximation function is denoted by $\hat{G}_k(x)$ as (\ref{G}). It can be validated that $\hat{G}_k(x)$ is concave and differentiable, whose derived function is denoted by $\hat{G}'_k(x)$ as (\ref{G'}).

\newcounter{mytempeqncnt}
\begin{figure*}[!t]
	\normalsize
	\setcounter{mytempeqncnt}{\value{equation}}
	\setcounter{equation}{22}
	
\begin{equation}\label{G}
\hat{G}_k(x)\!=\!
\begin{cases}
\frac{R(n)\!-\!R(n-1)}{4\delta}(\!x\!-\!n\!)^2\!+\!\frac{R(n)\!+\!R(n-1)}{2}(\!x\!-\!n\!)\!+ G_k(n)+\frac{\delta(R(n)-R(n-1))}{4}& n\!-\!\delta\!<\! x\!<\!n\!+\!\delta, n\!\in\!\mathcal{N}_+\\
G_k(x) & else
\end{cases}
\end{equation}

\begin{equation}\label{G'}
\hat{G}'_k(x)\!=\!
\begin{cases}
R(0)& 0\leq x\leq 1-\delta\\
R(n)& n+\delta\leq x \leq n+1-\delta,n\in\mathcal{N}_+\\
\frac{R(n)\!-\!R(n-1)}{2\delta}(\!x\!-\!n\!)+\frac{R(n)\!+\!R(n-1)}{2}&n\!-\!\delta\!<\! x\!<\!n\!+\!\delta, n\!\in\!\mathcal{N}_+
\end{cases}
\end{equation}
	\setcounter{equation}{\value{mytempeqncnt}}
	\hrulefill
\end{figure*}


Similarly, we approximate the rate function $S_k(p_k)$ by
$\hat{S}_k(p_k)=p_k \hat{G}_k(1/p_k)$. And its derived function is
\setcounter{equation}{24}
\begin{equation}
\hat{S}'_k(p_k)=\hat{G}_k(\frac{1}{p_k})-\frac{1}{p_k}\hat{G}'_k(\frac{1}{p_k}).
\end{equation}

When $\delta\rightarrow 0$, the solution to the approximated version of the optimization problem (\ref{newOP}) tends to be the real solution.

We use the gradient projection method to solve the approximated problem, which is an iterative method. In the $t$-th iteration step, we find a new value for CSI updating frequency in the direction of the gradient of $\sum_{j=1}^{K}\hat{S}_k(p_j)$:
\begin{equation}
x_k^{(t)}=p_k^{(t)}+a\frac{\hat{S}'_k(p_k^{(t)})}{\sqrt{\sum_{k=1}^{K}\hat{S}_k^{'2}(p_k^{(t)})}},
\end{equation} 
where $a$ is the iteration step size.

The point $\bm{x}=(x_1,x_2,\cdots,x_K)\in \mathcal{R}^K$ may lie beyond the constraint space $\mathcal{P}=\{\bm{p}|\sum_{k=1}^{K}p_k=T,0\leq p_k\leq 1,k=1,\cdots,K\}$, So we need to project $\bm{x}$ to the space $\mathcal{P}$. The projection is equivalent to solving the optimization problem
\begin{equation}\label{projection}
\begin{split}
\max_{\bm{p}^{(t+1)}} &||\bm{p}^{(t+1)}-\bm{x}^{(t)}||_2\\
s.t.&~~\bm{p}^{(t+1)}\in\mathcal{P}.
\end{split}
\end{equation}

By deriving the KKT conditions of (\ref{projection}), we get
\begin{equation}\label{KKT}
p_k=x_k+{\nu+\lambda_k-\mu_k},
\end{equation}
where $\lambda_i\geq 0,\mu_i\geq 0,\nu$ are the Lagrange multipliers of (\ref{projection}). If we arrange $\{x_i\}$ in ascending order as $\{x_{a(i)}\}$, then the elements of its projection vector $\{p_{a(i)}\}$ are also in ascending order. We maintain $i,j$ as the indexes of the first value larger than 0 and the last value less than 1 in $\{p_{a(i)}\}$, respectively. Summing up (\ref{KKT}) from $i$ to  $j$ gets
\begin{equation}
\nu=\frac{T-(K-j)-\sum_{l=i}^{j}x_l}{j-i+1}.
\end{equation}

We initialize $i=1,j=K$ and test whether the constraints $0\leq p_{a(k)}\leq 1$ are met. If not, we either increase $i$ or decrease $j$, according to the value of $x_{a(i)}+x_{a(j)}$. The detailed projection method is described by Algorithm \ref{alg:projection}.
\begin{algorithm}
	\caption{Projection Algorithm for (\ref{projection})}\label{alg:projection}
	\begin{algorithmic}[1]  
		\STATE Sort $\{x_k\}$ in ascending order as $\{x_{a(k)}\}$. \\
		Initialize $i=1,j=K$.
		\WHILE{$i<j$}
		\STATE $\nu=\frac{T-(K-j)-\sum_{l=i}^{j}x_l}{j-i+1}$
		\IF{$x_i+\nu\geq 0 ~\textrm{and}~ x_j+\nu\leq 1$}
		\FOR{$k=i,i+1,\cdots,j$}
		\STATE $p_{a(k)}=x_{a(k)}+\frac{T-(K-j)-\sum_{l=i}^{j}x_l}{j-i+1}$
		\ENDFOR
		\STATE Break.
		\ELSE
		\IF{$x_{a(i)}+x_{a(j)}\geq1$}
		\STATE $p_{a(j)}=1$
		\STATE $j=j-1$
		\ELSE
		\STATE $p_{a(i)}=0$
		\STATE $i=i+1$
		\ENDIF
		\ENDIF
		\ENDWHILE
	\end{algorithmic}
	
\end{algorithm} 

\subsection{Optimizing the Pilot Length}
The optimization problem (\ref{OP}) is finally solved by finding the maximum sum rate provided by the solution of (\ref{newOP}) for each value of $T\in\{0,1,\cdots,K\}$.

When $T=0$, the CSI updating frequencies of all users are $0$. As a result, the sum rate is $\sum_{k=1}^{K}S_k(p_k)=0$. In this case, the system spends no resource on channel estimation, therefore the precoding process is meaningless.

When $T=K$, the CSI updating frequencies of all users are $1$, which means the system applies conventional continuous estimation scheme and estimates the channels of the whole users in every block. The sum rate can be calculated as $(1-K/C)\sum_{k=1}^{K}R_k(0)$.
\section{Numerical Results}
\begin{figure}[!t]
	\centering
	\includegraphics[width=3.5in]{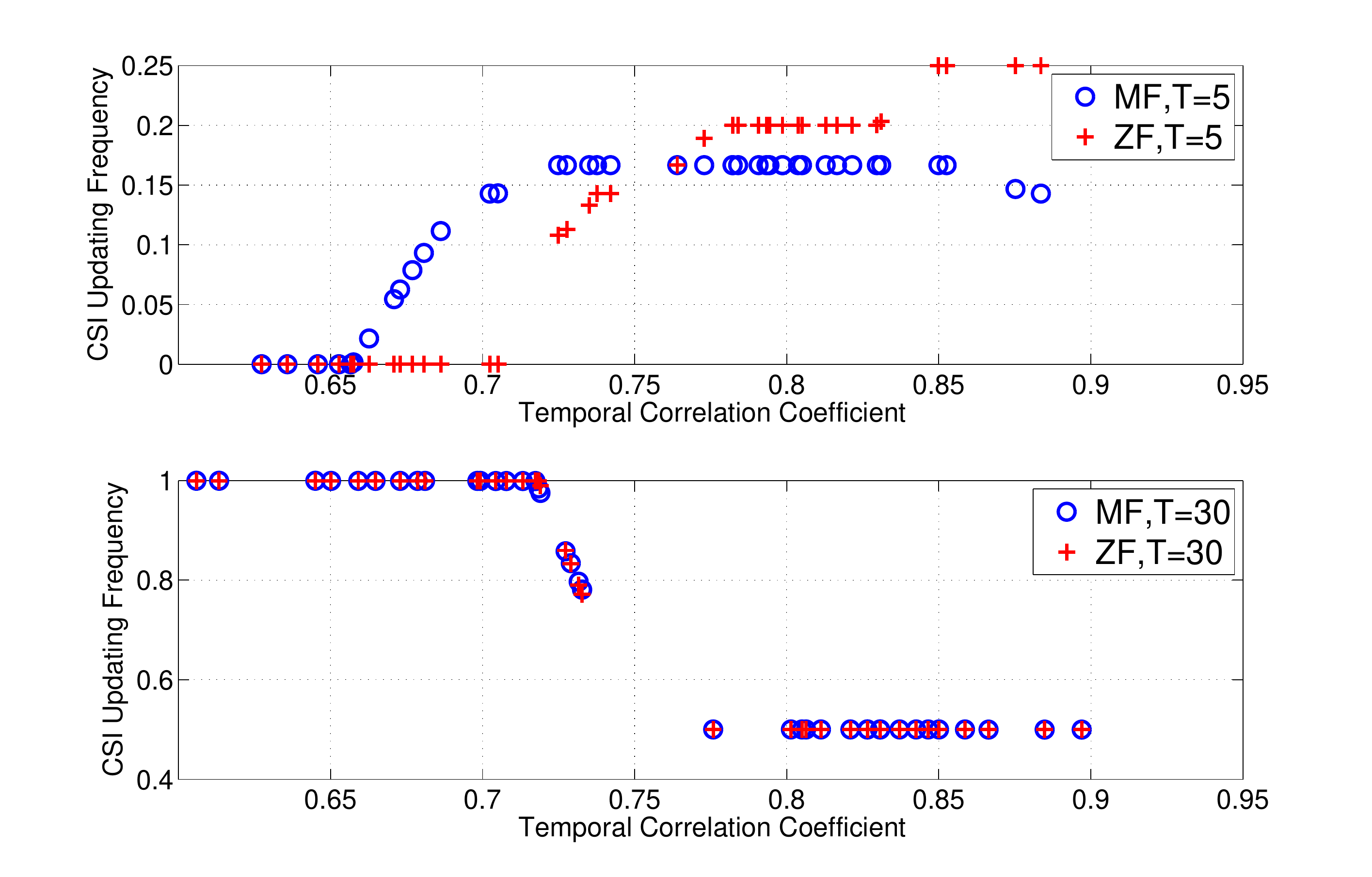}
	\caption{CSI updating frequency vs. temporal correlation coefficient under different pilot length. $M=64, K=40, C=50$, per user SNR is 10 dB.}
	\label{fig:Rho_p}
\end{figure}
\begin{figure}[!t]
	\centering
	\includegraphics[width=3.3in]{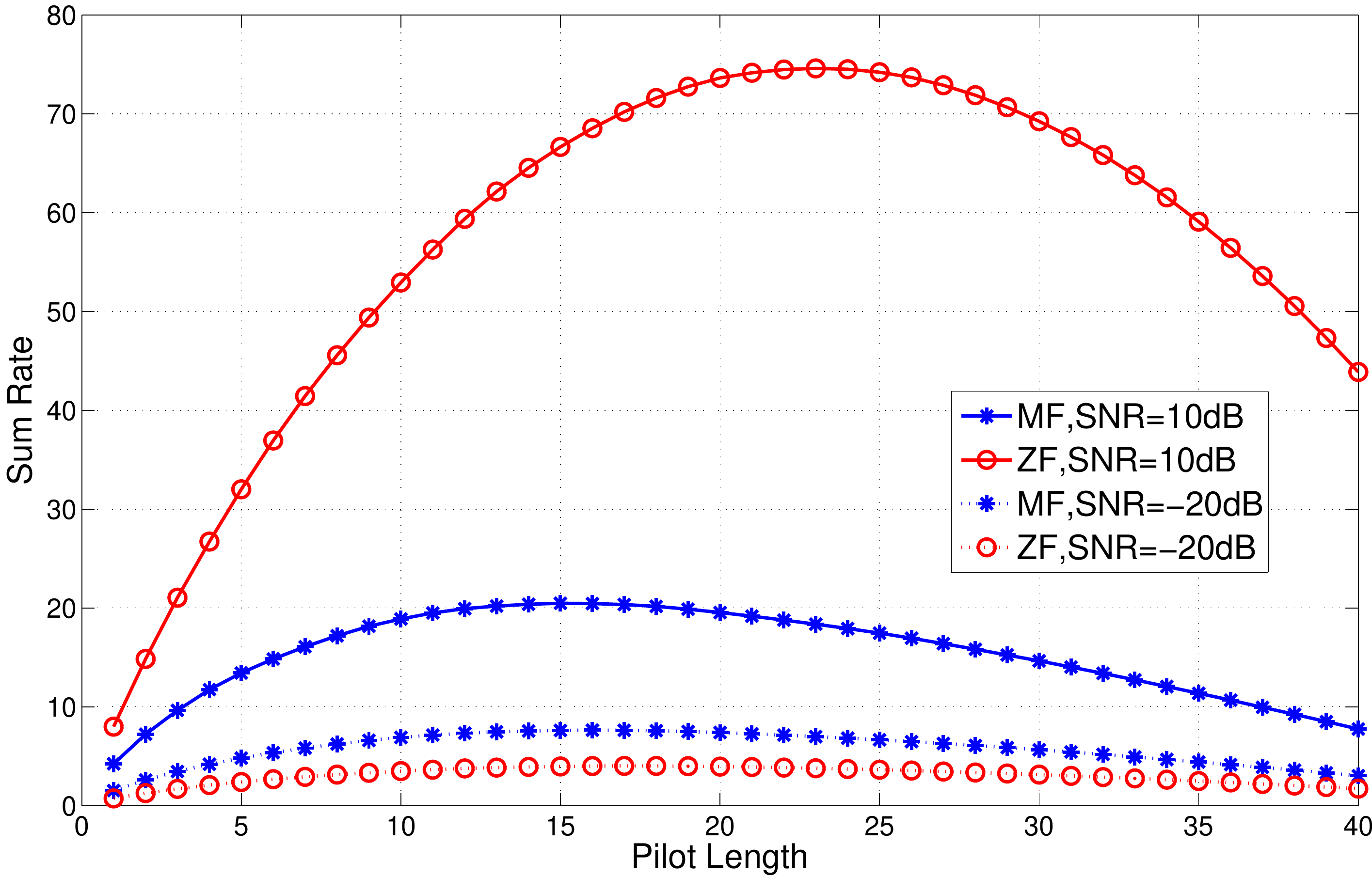}
	\caption{Sum rate vs. pilot length. $M=64, K=40, C=50$.}
	\label{fig:R_T}
\end{figure}
\begin{figure}[!t]
	\centering
	\includegraphics[width=3.5in]{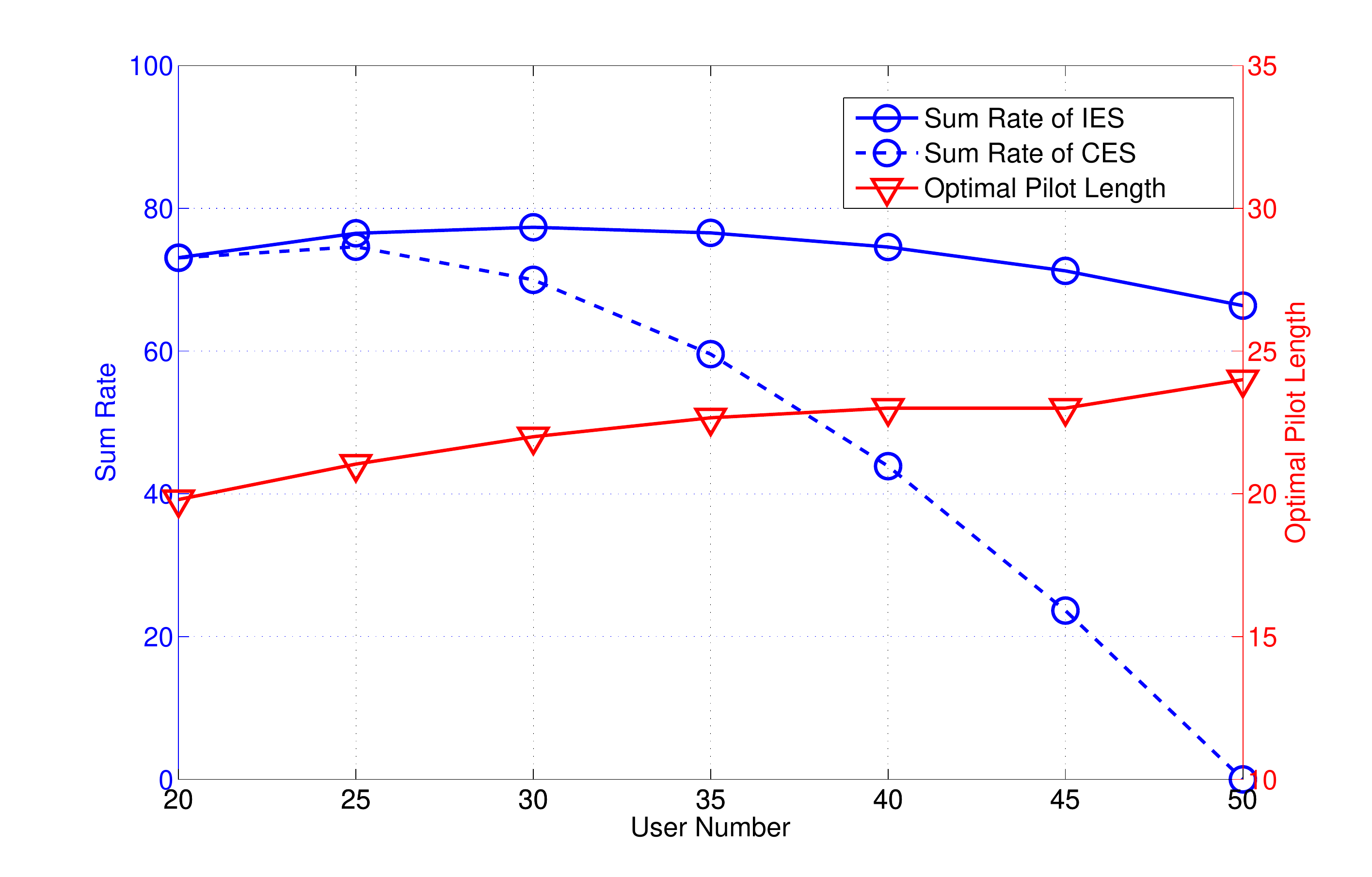}
	\caption{The sum rate and optimal pilot length of the proposed scheme using ZF precoding. $M=64, C=50$, per user SNR is 10 dB.}
	\label{fig:RT_K}
\end{figure}
In this section we evaluate the performance of the proposed estimation scheme.

We assume the same path loss coefficient for each user. The BS splits the total power equally to each user in order to ensure fairness. So the receive SNR of each user is the same. The temporal correlation coefficient of each user is uniformly distributed in the interval $[0.6,0.9]$.

Fig. \ref{fig:Rho_p} shows the allocation of estimation resource to users with different temporal correlation coefficients. When the total available pilot length is extremely low ($T=5$), the proposed scheme allocates zero estimation resource to users with temporal correlation coefficients below a certain threshold to leave the limited resource to users with better temporal correlation. The threshold of ZF precoder is higher than the threshold of the MF precoder, because it suffers more severely from the aged CSI. On the other hand, when the available pilot length is sufficiently high ($T=30$), both MF and ZF precoder estimate allocation the maximum estimation resource to the users with poor temporal correlation by setting the CSI updating frequency to 1.

The relationship between sum rate and pilot length is depicted in Fig. \ref{fig:R_T}. The sum rate of ZF outperforms that of MF in high SNR region and loses its advantage in low SNR region. The optimal pilot length of ZF is larger than that of MF under the same receive SNR. The sum rate of the conventional continuous estimation scheme is achieved at the point where pilot length is $K=40$, and our proposed scheme improves the performance by about 74\% using ZF precoder when the per user receive SNR is 10 dB.

Finally we investigate the influence of user number on the proposed scheme using ZF precoding, which is shown in Fig. \ref{fig:RT_K}. The results using MF precoding are omitted due to similarity. As the user number grows, the sum rate of the proposed scheme (IES) experiences the process of first increasing and then declining, which means the system first enjoys user diversity and then becomes overwhelmed by the excessive users. The performance gap between our scheme and the continuous estimation scheme (CES) increases with the growing number of users. When user number equals the length of channel block, CES fails to transmit any data while our scheme still obtains about 86\% of the peak sum rate. On the other hand, the optimal pilot length increases all the time, but its growth rate is much lower than the growth rate of user number. 
\section{Conclusion}\label{sec:conclusion}
In this paper, we investigate the influence of the age of CSI on the channel capacity of massive MIMO systems and derive closed-form expressions for transmission rate under two linear precoders, namely matched filter and zero forcing. The intermittent estimation scheme is proposed to reduce the overhead of CSI acquisition. We show the CSI updating intervals should follow a quasi-periodic distribution and obtain the optimal CSI updating frequency for each user to maximize sum capacity in the limited block length. The numerical results shows a great performance gain of the proposed scheme compared to the conventional continuous estimation scheme.
%



%

\appendices
\section{Proof of Theorem 1}\label{appendices1}
According to the Gaussian-Markov process, we know the correlation between the estimated channel and the channel of ground truth reads
\begin{equation}
\bm{h}_k=\rho_k^{n}\bm{\hat{h}}_k+\sqrt{1-\rho_k^{2n}}\bm{e}_k.
\end{equation}

When the MF precoding scheme is applied, the random variable $\bm{\hat{h}}_k^H\bm{v}_k$ satisfies the chi square distribution with $2M$ degrees of freedom by a scale factor of $1/(2\eta_k^\textrm{MF})$, so we can obtain the expectation of effective channel, the variation of effective channel and the interference as
\begin{equation}\label{MF1}
\epsilon_k\beta_k|\mathbb{E}\bm{h}_k^T\bm{v}_k|^2=\epsilon_k\beta_k M\rho_k^{2n},
\end{equation}
\begin{equation}\label{MF2}
\epsilon_k\beta_k(\mathbb{E}|\bm{h}_k^T\bm{v}_k|^2-|\mathbb{E}\bm{h}_k^T\bm{v}_k|^2)= \epsilon_k\beta_k,
\end{equation}
\begin{equation}\label{MF3}
\sum_{i\neq k}\epsilon_i\beta_i\mathbb{E}|\bm{h}_k^T\bm{v}_i|^2=\sum_{i\neq k}\epsilon_i\beta_i,
\end{equation}
respectively. We get the expression for SINR under MF precoding scheme by substituting (\ref{MF1}), (\ref{MF2}) and (\ref{MF3}) into (\ref{SINR}).

Similarly, we calculate the expectation of effective channel, the variation of effective channel and the interference under ZF precoding scheme as
\begin{equation}\label{ZF1}
\epsilon_k\beta_k|\mathbb{E}\bm{h}_k^T\bm{v}_k|^2=\epsilon_k\beta_k (M-K)\rho_k^{2n},
\end{equation}
\begin{equation}\label{ZF2}
\epsilon_k\beta_k(\mathbb{E}|\bm{h}_k^T\bm{v}_k|^2-|\mathbb{E}\bm{h}_k^T\bm{v}_k|^2)= \epsilon_k\beta_k(1-\rho_k^{2n}),
\end{equation}
\begin{equation}\label{ZF3}
\sum_{i\neq k}\epsilon_i\beta_i\mathbb{E}|\bm{h}_k^T\bm{v}_i|^2=(1-\rho_k^{2n})\sum_{i\neq k}\epsilon_i\beta_i,
\end{equation}
respectively. and the SINR is derived by substituting (\ref{ZF1}), (\ref{ZF2}) and (\ref{ZF3}) into (\ref{SINR}).
\section{Proof of Theorem 2}\label{appendices2}
	Firstly, we show that the function $G_k(x)$ is concave in its domain of definition $\mathcal{D}=\{x\in \mathcal{R}|x\geq 0\}$.
	
	Assume $x_1, x_2, x_3$, are three arbitrary different real numbers in $\mathcal{D}$. We arbitrarily choose $\theta\in[0,1]$ and get $x_3=\theta x_1+(1-\theta)x_2$. Without loss of generality, we assume $x_1<x_2<x_3$. 
	
	According to the definition of function $G_k(x)$, we have
	\begin{equation}
	G_k(x_i)=(1-t_i)G_k(\lfloor x_i \rfloor)+t_i G_k(\lfloor x_i \rfloor+1),~i=1,2,3,
	\end{equation}
	where $t_i= x_i-\lfloor x_i \rfloor,i=1,2,3$. 
	
	When $\lfloor x_1 \rfloor=\lfloor x_2 \rfloor$, we have
	\begin{equation}
	\frac{G_k(x_2)\!-\!G_k(x_1)}{x_2-x_1}=R_k(\lfloor x_2 \rfloor).
	\end{equation}
	
	When $\lfloor x_1 \rfloor<\lfloor x_2 \rfloor$, we can get the following inequality due to the monotonicity of $R_k(i)$:
	\begin{equation}
	\begin{split}
	&\frac{G_k(x_2)\!-\!G_k(x_1)}{x_2-x_1}\\
	\!=&\!\frac{t_1 R_k(\lfloor x_1 \rfloor)\!+\!\sum_{i=\lfloor x_1 \rfloor+1}^{\lfloor x_2 \rfloor-1}R_k(i)\!+\!(1-t_2)R_k(\lfloor x_2 \rfloor))}{x_2-x_1}\\
	\geq& R_k(\lfloor x_2 \rfloor)).
	\end{split}
	\end{equation}
	
	Similarly, we obtain
	\begin{equation}
	\frac{G_k(x_3)-G_k(x_2)}{x_3-x_2}\leq R_k(\lfloor x_2 \rfloor).
	\end{equation}
	
	So the inequality
	\begin{equation}
	\frac{G_k(x_3)-G_k(x_2)}{x_3-x_2} \leq \frac{G_k(x_2)-G_k(x_1)}{x_2-x_1}
	\end{equation}
	holds, which proves the concavity of the function $G_k(x)$.
	
	Then, by \emph{Jensen's inequality}\cite{Boyd2004convex}, the transmission rate of the $k$-th user satisfies the inequality
	\begin{equation}
	\mathbb{E}\{R_k\}=p_k\sum_n f_{k,n}G_k(n)\leq p_k G_k\left(\sum_n f_{k,n}\right)=p_k G\left(\frac{1}{p_k}\right),
	\end{equation}
	where the equality holds when $\left\{f_{k,n}\right\}$ satisfies (\ref{fkn}).
\section{Proof of Theorem 3}\label{appendices3}
	We only need to prove the function $S_k(p_k)=p_k G_k(1/{p_k})$ is concave with respect to $p_k$.
	
	Arbitrarily select two real numbers $p_{k1},p_{k2}$ from the interval $[0,1]$ , for any $\theta\in[0,1]$ we have 
	\begin{equation}
	\begin{split}
	&G_k\left(\frac{\theta p_{k1}}{\theta p_{k1}\!+\!\bar{\theta} p_{k2}}\frac{1}{p_{k1}}\!+\!\frac{\theta p_{k2}}{\theta p_{k1}\!+\!\bar{\theta} p_{k2}}\frac{1}{p_{k2}}\right)\\
	\leq&\frac{\theta p_{k1}}{\theta p_{k1}\!+\!\bar{\theta} p_{k2}}G_k\left(\frac{1}{p_{k1}})\!+\!\frac{\theta p_{k2}}{\theta p_{k1}\!+\!\bar{\theta} p_{k2}}G_k(\frac{1}{p_{k2}}\right),
	\end{split}
	\end{equation}
	due to the concavity of $G_k(x)$ from the proof of Theorem 2, where $\bar{\theta}=1-\theta$.
	
	The above inequality can be simplified as
	\begin{equation}
	\begin{split}
	&(\theta p_{k1}+\bar{\theta} p_{k2})G_k\left(\frac{1}{\theta p_{k1}+\bar{\theta} p_{k2}}\right)\\
	\leq &\theta p_{k1}G_k\left(\frac{1}{p_{k1}}\right)+\bar{\theta} p_{k2}G_k\left(\frac{1}{p_{k2}}\right),
	\end{split}
	\end{equation}
	which proves the concavity of the function $S_k(p_k)$.

\section*{Acknowledgment}
This work is sponsored in part by the Nature Science Foundation of China (No. 61461136004, No. 61571265, No. 91638204), and Hitachi Ltd.



 \bibliographystyle{ieeetr}
 \bibliography{IEEEabrv,myref}

\end{document}